\documentclass[conference]{IEEEtran}
\usepackage{epsfig,epstopdf} 
\usepackage{amsmath,amssymb,amsthm,cite,graphicx,array}
\theoremstyle{definition}
\newtheorem{theorem}{Theorem}
\newtheorem{lemma}{Lemma}

\newtheorem{note}{Note}
\newtheorem{example}{Example}
\usepackage{color}

\title{Reduced Complexity Index Codes and Improved Upperbound on Broadcast Rate for Neighboring Interference Problems}

\begin{document}

\author{Mahesh~Babu~Vaddi~and~B.~Sundar~Rajan\\ 
 Department of Electrical Communication Engineering, Indian Institute of Science, Bengaluru 560012, KA, India \\ E-mail:~\{vaddi,~bsrajan\}@iisc.ac.in}
 
\maketitle
\begin{abstract}
A single unicast index coding problem (SUICP) with symmetric neighboring interference (SNI) has $K$ messages and $K$ receivers, the $k$th receiver $R_{k}$ wanting the $k$th message $x_{k}$ and having the interference with $D$ messages after and $U$ messages before its desired message. Maleki, Cadambe and Jafar studied SUICP(SNI) because of its importance in topological interference management problems. Maleki \textit{et. al.} derived the lowerbound on the broadcast rate of this setting to be $D+1$. In our earlier work, for SUICP(SNI) with arbitrary $K,D$ and $U$, we defined set $\mathbf{S}$ of 2-tuples and for every $(a,b) \in \mathbf{S}$, we constructed $b$-dimensional vector linear index code with rate $D+1+\frac{a}{b}$ by using an encoding matrix of dimension $Kb \times (b(D+1)+a)$. In this paper, we use the symmetric structure of the SUICP(SNI) to reduce the size of encoding matrix by partitioning the message symbols. The rate achieved in this paper is same as that of the existing constructions of vector linear index codes. More specifically, we construct $b$-dimensional vector linear index codes for SUICP(SNI) by partitioning the $Kb$ messages into $b(U+1)+c$ sets for some non-negative integer $c$. We use an encoding matrix of size $\frac{Kb}{b(U+1)+c} \times \frac{b(D+1)+a}{b(U+1)+c}$ to encode each partition separately. The advantage of this method is that the receivers need to store atmost $\frac{b(D+1)+a}{b(U+1)+c}$ number of broadcast symbols (index code symbols) to decode a given wanted message symbol. We also give a construction of scalar linear index codes for SUICP(SNI) with arbitrary $K,D$ and $U$. We give an improved upperbound on the braodcast rate of SUICP(SNI).

\end{abstract}
\section{Introduction and Background}
\label{sec1}

\IEEEPARstart {A}{n} index coding technique uses the braodacast nature of wireless medium to optimise the number of transmissions. An index coding problem, comprises a transmitter that has a set of independent messages and a set of receivers. Each receiver knows a subset of messages, called its side-information, and demands another subset of messages, called its want-set. The transmitter can take cognizance of the side-information of the receivers and broadcast coded messages, called the index code, over a noiseless channel. The objective is to minimize the number of coded transmissions, called the length of the index code, such that each receiver can decode its wanted messages using its side-information and the coded messages. 

The problem of index coding with side-information was introduced by Birk and Kol \cite{ISCO}. Ong and Ho \cite{OnH} classified the index coding problem depending on the demands and the side-information possessed by the receivers. An index coding problem is single unicast if the demand-sets of the receivers are disjoint and the cardinality of demand-set of every receiver is one.

A single unicast index coding problem (SUICP) has $K$ messages $\{x_0,x_1 ,\ldots ,x_{K-1}\}$ and $K$ receivers $\{R_0 ,R_1$, $\ldots, R_{K-1}\}$ for some positive integer $K$. Receiver $R_k$ wants the message $x_k$ and knows a subset of messages in $\{x_0,x_1 ,\ldots,x_{K-1}\}$ as side-information. We assume that the messages belongs to a finite alphabet $\mathcal{B}$. In a single unicast index coding problem, the side-information is represented by a directed graph $G$ = ($V$,$E$) with $V=\{x_0,x_1,\ldots,x_{K-1}\}$ vertices and $E$ is the set of edges such that the directed edge $(x_i,x_j)\in E$ if receiver $R_{i}$ knows $x_{j}$. This graph $G$ for a given index coding problem is called the side-information graph. In \cite{YBJK}, Bar-Yossef \textit{et. al.} studied single unicast index coding problems. In this paper, we use $\mathcal{W}_k$ to denote want set and $\mathcal{K}_k$ to denote side-information of the receiver $R_k$. The messages which are neither wanted by nor known to $R_k$ is called interference $\mathcal{I}_k$ to $R_k$. 

The solution (includes both linear and nonlinear) of the index coding problem must specify a finite alphabet $\mathcal{B}_P$ to be used by the transmitter, and an encoding scheme $\varepsilon:\mathcal{B}^{K} \rightarrow \mathcal{B}_{P}$ such that every receiver is able to decode the wanted message from the $\varepsilon(x_0,x_1,\ldots,x_{K-1})$ and the known information. The minimum encoding length $l=\lceil log_{2}|\mathcal{B}_{P}|\rceil$ for messages that are $t$ bit long ($\vert\mathcal{B}\vert=2^t$) is denoted by $\beta_{t}(G)$. The broadcast rate of the index coding problem with side-information graph $G$ is defined \cite{ICVLP} as,
\begin{align*}
\beta(G) \triangleq   \inf_{t} \frac{\beta_{t}(G)}{t}.
\end{align*}

If $t = 1$, it is called scalar broadcast rate. For a given index coding problem, the broadcast rate $\beta(G)$ is the minimum number of index code symbols required to transmit to satisfy the demands of all the receivers. The capacity $C(G)$ for the index coding problem is defined as the maximum number of message symbols transmitted per index code symbol such that every receiver gets its wanted message symbols and all the receivers get equal number of wanted message symbols. The broadcast rate and capacity are related as \cite{MCJ}
\begin{center}	
$C(G)=\dfrac{1}{\beta(G)}$.
\end{center}  

In a vector linear index code $x_k=(x_{k,1},x_{k,2},\ldots,x_{k, p_k }) \in\mathbb{F}_q^{p_k},~x_{k,j} \in \mathbb{F}_q$ for 
$k \in [0:K-1]$ and $j \in [1:p_k]$ where $\mathbb{F}_q$ is a finite field with $q$ elements. In vector linear index coding setting, we refer $x_k \in \mathbb{F}_q^{p_k}$ as a message-vector or message and $x_{k,1},x_{k,2},\ldots,x_{k,p_k} \in \mathbb{F}_q$ as the message symbols. An index coding is a mapping defined as
\begin{align*}
\mathfrak{E}: \mathbb{F}^{p_0+p_1+\ldots+p_{K-1}}_q \rightarrow \mathbb{F}^N_q,
\end{align*}
where $N$ is the length of index code. The index code $\mathfrak{C}=\{(c_0,c_1,\ldots,c_{N-1}) \}$ is the collection of all images of the mapping $\mathfrak{E}$. We call the symbols $c_0$,$c_1,\ldots$,$c_{N-1}$ the code symbols,  which are the symbols broadcasted by the transmitter. If $p_0=p_1=\cdots=p_{K-1}=p$, then the index code is called symmetric rate $p$-dimensional vector linear index code. If $p_0=p_1=\cdots=p_{K-1}=1$, then the index code is called scalar index code.

A p-dimensional vector linear index code of length $N$ is represented by a matrix $\mathbf{L}$ $(\in \mathbb{F}_q^{pK \times N})$, where the $j$th column contains the coefficients of the $j$th coded transmission and the $(ip+j)$th row $L_{ip+j}$ $(\in \mathbb{F}_q^{1\times N})$ contains the coefficients used for mixing message $x_{i,j}$ in the $N$ transmissions for every $i \in [0:K-1]$ and $j \in [1:p]$. The broadcast vector is 
\begin{align*}
[c_0,c_1,\ldots,c_{N-1}]&=[\underbrace{x_{0,1},\ldots,x_{0,p}}_{x_0}\ldots \underbrace{x_{K-1,1},\ldots,x_{K-1,p}}_{x_{K-1}}]\mathbf{L}\\&=\sum_{i=0}^{N-1}\sum_{j=1}^{p}x_{i,j}L_{ip+j}.
\end{align*}

Example \ref{introex2} given below illustrates the advantage of vector linear index codes. 
\begin{example}
\label{introex2}
Consider the index coding problem with wanted message and side-information as given in Table below.
\begin{table}[ht]
\centering
\begin{tabular}{|c|c|c|}
\hline
$R_k$ & \textbf{$\mathcal{W}_k$} & \textbf{$\mathcal{K}_k$} \\
\hline
$R_0$ & $x_0$ & $x_1,x_3$\\
\hline
$R_1$ & $x_1$ & $x_2,x_3$  \\
\hline  
$R_2$ & $x_2$ & $x_0$\\
\hline
$R_3$ & $x_3$ & $x_1,x_2$  \\
\hline  
\end{tabular}
\end{table}

The side-information graph of this SUICP is given in Fig. \ref{ex1}.
\begin{figure}
\centering
\includegraphics[scale=0.45]{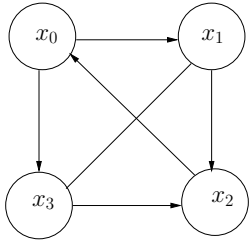}\\
\caption{side-information graph of SUICP given in Example \ref{introex2}}
\label{ex1}
\end{figure}
Let $x_{k,1},x_{k,2}$ be the two generations of the message symbol $x_k$ for $k \in [0:3]$.  A vector linear index code with symmetric rate $\frac{2}{5}$ for this index coding problem is 
\begin{align*}
(y_0,y_1,y_2,y_3,y_4)=\{&x_{0,1}+x_{1,1},~~x_{1,1}+x_{2,1},\\&x_{0,2}+x_{3,1},~~x_{3,1}+x_{2,2},\\&x_{1,2}+x_{3,2}\}.
\end{align*}

It is easy to see that from the five broadcast symbols $(y_0,y_1,y_2,y_3,y_4)$, every receiver can decode its two wanted message symbols by using the available side-information with them. 
\begin{note}
The optimal scalar linear index code length of the index coding problem given in Example \ref{introex2} is three. That is, atleast three transmissions are required by using scalar linear index code to receive one message symbol per receiver. However, by using two dimensional vector linear index code given in Example \ref{introex2}, every receiver receives one message symbol by using 2.5 transmissions.
\end{note}
\end{example}

The general form of the AIR matrix of size $m \times n$ for any given positive integers $m$ and $n(m \geq n)$ is shown in Fig. \ref{fig1}. It consists of several submatrices of different sizes as shown in Fig.\ref{fig1}. The description of the submatrices are as follows: Let $a$ and $b$ be two positive integers and $b$ divides $a$. The following matrix denoted by $\mathbf{I}_{a \times b}$ is a rectangular matrix
\begin{align}
\label{rcmatrix}
\mathbf{I}_{a \times b}=\left.\left[\begin{array}{*{20}c}
   \mathbf{I}_b  \\
   \mathbf{I}_b  \\
   \vdots  \\
   \mathbf{I}_b 
   \end{array}\right]\right\rbrace \frac{a}{b}~\text{number~of}~ \mathbf{I}_b~\text{matrices}
\end{align}
and $\mathbf{I}_{b \times a}$ is the transpose of $\mathbf{I}_{a \times b}.$ 

For a given $m$ and $n,$  let $m-n= \lambda_0$ and
\begin{align}
\nonumber
n&=\beta_0 \lambda_0+\lambda_1, \nonumber \\
\lambda_0&=\beta_1\lambda_1+\lambda_2, \nonumber \\
\lambda_1&=\beta_2\lambda_2+\lambda_3, \nonumber \\
&~~~~~~\vdots \nonumber \\
\lambda_i&=\beta_{i+1}\lambda_{i+1}+\lambda_{i+2}, \nonumber \\ 
&~~~~~~\vdots \nonumber \\ 
\lambda_{l-1}&=\beta_l\lambda_l
\label{chain}
\end{align}
where $\lambda_{l+1}=0$ for some integer $l,$ $\lambda_i,\beta_i$ are positive integers and $\lambda_i < \lambda_{i-1}$ for $i=1,2,\ldots,l$.

\begin{figure*}
\centering
\includegraphics[scale=0.36]{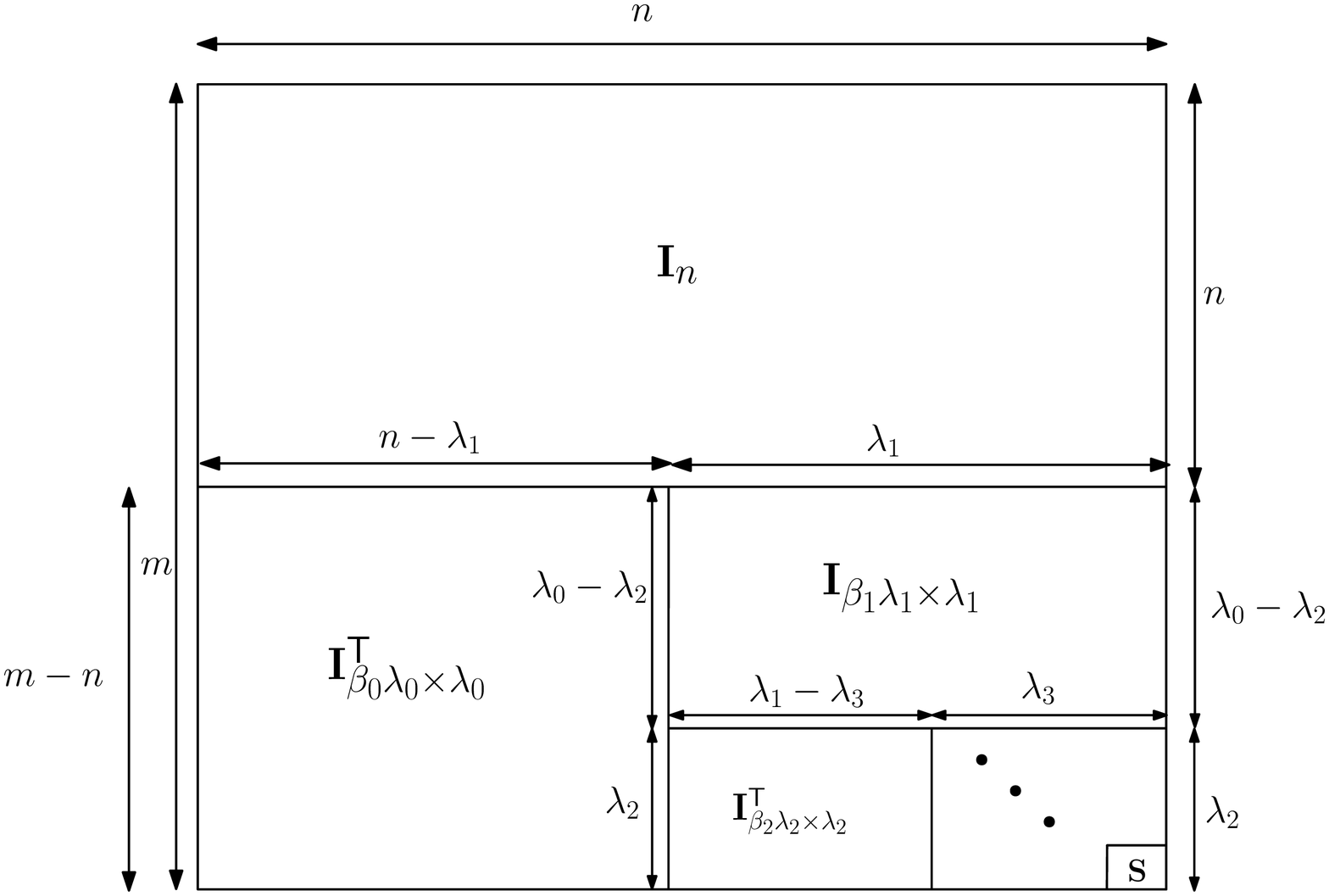}\\
~ $\mathbf{S}=\mathbf{I}_{\lambda_{l} \times \beta_l \lambda_{l}}$ if $l$ is even and ~$\mathbf{S}=\mathbf{I}_{\beta_l\lambda_{l} \times \lambda_{l}}$ otherwise.
\caption{AIR matrix of size $m \times n~(m \geq n)$.}
\label{fig1}
\end{figure*}


\subsection{Symmetric neighboring interference single unicast index coding problem}
A symmetric neighboring interference single unicast index coding problem (SUICP(SNI)) with $K$ messages and $K$ receivers, each receiver has a total of $U+D<K$ interfering messages, corresponding to the $D$ messages after and $U$ messages before its desired message. In this setting, the $k$th receiver $R_{k}$ demands the message $x_{k}$ having the interference
\begin{equation}
\label{antidote}
{I}_k= \{x_{k-U},\dots,x_{k-2},x_{k-1}\}\cup\{x_{k+1}, x_{k+2},\dots,x_{k+D}\}.
\end{equation}

The side-information of this setting is given by
\begin{align}
\label{sideinformation}
\mathcal{K}_{k}=(\mathcal{I}_{k} \cup x_{k})^c.
\end{align}

All the subscripts in SUICP(SNI) are to be considered $~\text{\textit{modulo}}~ K$. 
\subsection{Motivating Example}
The motivation for studying SUICP(SNI) can be understand from the topological interference management problem given in Fig.\ref{cell}. In this setting, there are five base stations (shown as black squares) and each base station wants to transmit a message to one receiver (shown as white squares). In this setting, every receiver sees interference from its neighboring base stations. This topological interference management problem can be modelled as SUICP(SNI) with $K=5,D=1$ and $U=1$.

\begin{figure}[h]
\centering
\includegraphics[scale=0.60]{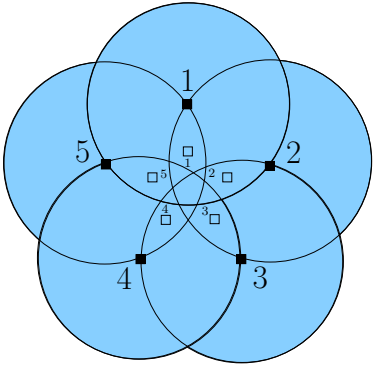}\\
\caption{SUICP(SNI) corresponding to TIM problem}
\label{cell}
\end{figure}

\subsection{Existing Results}

Maleki \textit{et al.} \cite{MCJ} found the capacity of SUICP(SNI) with $K\rightarrow \infty$. The capacity of SUICP(SNI) with $K\rightarrow \infty$ is 
\begin{align}
\label{capacityint}
C=\frac{1}{D+1}~\text{per~message}.
\end{align}

Maleki \textit{et al.} \cite{MCJ} proved the outerbound for the capacity of SUICP(SNI) for finite $K$. The outerbound for the finite $K$ is given by
\begin{align}
\label{outerbound}
C \leq \frac{1}{D+1}.
\end{align}

Blasiak \textit{et al.} \cite{ICVLP} found the capacity of SUICP(SNI) with $U=D=1$ by using linear programming bounds. The capacity of this SUICP(SNI) with $U=D=1$ is given by
\begin{align}
\label{cap5}
\frac{\left\lfloor \frac{K}{2}\right\rfloor}{K}. 
\end{align}

In \cite{VaR2}, we give a construction of binary matrices with a given size $m \times n~(m > n)$, such that any $n$ adjacent rows in the matrix are linearly independent over every field $\mathbb{F}_q$. We call these matrices as Adjacent Independent Row (AIR) matrices. In \cite{VaR5}, for SUICP(SNI) with arbitrary $K,D$ and $U$, we define a set $\mathcal{\mathbf{S}}$ of $2$-tuples as given below
\begin{align}
\label{set}
\mathbf{S}=\{(a,b):\text{gcd}(bK,b(D+1)+a)\geq b(U+1)\}
\end{align}
and we show that for every $(a,b) \in \mathcal{\mathbf{S}}$, the rate $D+1+\frac{a}{b}$ is achievable by using $b$-dimensional vector linear index codes and AIR matrices of size $Kb \times b(D+1)+a$. 

Jafar \cite{TIM} established the relation between index coding problem and topological interference management problem. The SUICP(SNI) is motivated by topological interference management problems. The capacity and optimal coding results in index coding can be used in corresponding topological interference management problems.
\subsection{Contributions}
The contributions of this paper are summarized below:
\begin{itemize}
\item We construct $b$-dimensional vector linear index codes for SUICP(SNI) by partitioning the $Kb$ messages into $b(U+1)+c$ sets for some non-negative integer $c$ satisfying the condition $\text{gcd}(Kb,b(D+1)+a)=b(U+1)+c$. We use $\frac{Kb}{b(U+1)+c} \times \frac{(b(D+1)+a)}{b(U+1)+c}$ size matrix to encode each partition separately. The advantage of this method is that the receivers need to store atmost $\frac{(b(D+1)+a)}{b(U+1)+c}$ received broadcast symbols to decode a given wanted message symbol.
\item This proposed index code construction identifies the receivers which will be able to decode their wanted messages instantly
(without using buffers).
\item We give a construction of scalar linear index codes for SUICP(SNI) with arbitrary $K,D$ and $U$. 
\item We give an improved upperbound on the braodcast rate of SUICP(SNI).
\end{itemize}

The paper is organized as follows. In Section \ref{sec3}, we give reduced complexity encoding for SUICP(SNI) by partitioning the message symbols. In Section \ref{sec6}, we give a scalar linear construction of index codes for SUICP(SNI). In Section \ref{sec7}, we give an improved upperbound on broadcast rate of SUICP(SNI). We summarize the paper in Section \ref{sec4}.

\section{Partition Based Vector Linear Index Codes for SUICP(SNI) with arbitrary K,U and D}
\label{sec3}

In this section, we use partition of message symbols to obtain a simple linear index code for SUICP(SNI) with arbitrary $K$, $D$ and $U$. Let
\begin{align}
\label{set11}
\mathbf{S}=\{(a,b):\text{gcd}(bK,b(D+1)+a)\geq b(U+1)\}
\end{align}
In a $b$-dimensional vector linear index code, $x_k=(x_{k,1}~x_{k,2}\ldots x_{k,b}) \in \mathbb{F}^{b}_q$ for $k \in [0:K-1]$. Receiver $R_k$ wants to decode the $b$ message symbols $x_{k,1}~x_{k,2}\ldots x_{k,b}$ for every $k \in [0:K-1]$. In the $b$-dimensional vector linear index code construction for SUICP(SNI) given in \cite{VaR5}, the $Kb$ messages corresponding to $K$ receivers are linearly combined to give $b(D+1)+a$ index code symbols by using an AIR matrix of size $Kb \times (b(D+1)+a)$. In Theorem \ref{novcccthmv1}, we partition the $Kb$ message symbols into $b(U+1)+c$ sets for some non-negative integer $c$ such that each set comprises of $\frac{Kb}{b(U+1)+c}$ messages. Then, by using the symmetry in the SUICP(SNI), we show that we can use an AIR matrix of size $\frac{Kb}{b(U+1)+c} \times \frac{b(D+1)+a}{b(U+1)+c}$ to encode each partition separately.
\begin{theorem}
\label{novcccthmv1}
Consider a SUICP(SNI) with arbitrary $K,D$ and $U$. Let $(a,b) \in \mathbf{S}$ and $\text{gcd}(bK,b(D+1)+a)=b(U+1)+c$ for some $c \in Z_{\geq 0}$. Let $\frac{Kb}{b(U+1)+c}=t$ and $\frac{b(D+1)+a}{b(U+1)+c}=\gamma$. An index code for this SUICP(SNI) with rate $D+1+\frac{a}{b}$ is obtained by using a $t \times \gamma$ AIR matrix. 
\end{theorem}
\begin{proof}
Consider a $b$-dimensional vector linear index coding. In a $b$-dimensional vector linear index code, $x_k \in \mathbb{F}^{b}_q$ for $k \in [0:K-1]$. In a $b$-dimensional vector linear index coding, receiver $R_k$ wants the $b$ message symbols $x_{k,1}~x_{k,2}\ldots x_{k,b}$. Let 
\begin{align*}
\mathbf{x}=[\underbrace{x_{0,1}x_{0,2}\ldots x_{0,b}}_{x_0}~\underbrace{x_{1,1}x_{1,2}\ldots x_{1,b}}_{x_1}\ldots \underbrace{x_{K-1,1}\ldots x_{K-1,b}}_{x_{K-1}}].
\end{align*}

In this index coding problem, we have 
\begin{align*}
\text{gcd}(bK,b(D+1)+a)=b(U+1)+c,
\end{align*}
for some $c \in Z_{\geq 0}$. Hence, $b(U+1)+a$ divides both $Kb$ and $b(D+1)+a$. Let 
\begin{align*}
\tau=b(U+1)+c.
\end{align*}

Let 
\begin{align*}
y_{k}=x_{\left \lfloor \frac{k}{b} \right \rfloor,k~\text{mod}~b+1}
\end{align*}
for $k \in [0:bK-1]$. We have 
\begin{align*}
&[\underbrace{y_0 y_1\ldots y_{b-1}}_{x_0}~\underbrace{y_{b}y_{b+1}\ldots y_{2b-1}}_{x_1}\ldots \underbrace{y_{(K-1)b}\ldots y_{Kb-1}}_{x_{K-1}}]\\&=[\underbrace{x_{0,1}x_{0,2}\ldots x_{0,b}}_{x_0}~\underbrace{x_{1,1}x_{1,2}\ldots x_{1,b}}_{x_1}\ldots \underbrace{x_{K-1,1}\ldots x_{K-1,b}}_{x_{K-1}}].
\end{align*}

Let $\mathcal{A}_i$ be the subset of $t$ messages as given below.
\begin{align*}
\mathcal{A}_i=\{y_i,y_{i+\tau},y_{i+2\tau},\ldots,y_{i+(t-1)\tau}\},
\end{align*}
for $i \in [0:\tau-1]$.

We have 
\begin{align*}
\mathcal{A}_0 \cup \mathcal{A}_1 \cup \ldots \cup \mathcal{A}_{\tau-1}=\{y_0,y_1,\ldots,y_{bK-1}\}.
\end{align*}

That is, the sets $\mathcal{A}_i$ for $i \in [0:\tau-1]$ partition the $bK$ messages into $\tau$ disjoint sets. For every $s \in [\gamma:t-1]$, we have $s\tau \geq b(D+1)+a$. That is, for every $j \in [0:t-1]$ and $s \in [\gamma:t-1]$, the receiver wanting the message symbol $y_{i+j\tau}$ knows $y_{i+j\tau+s\tau}$. 

Let receiver $R_k$ wants the message $y_{i+j\tau}$ in $\mathcal{A}_i$. The interfering and known messages for $R_{k}$ in $\mathcal{A}_i$ is shown in \eqref{novcmdsv3}.
\begin{align}
\label{novcmdsv3}
\nonumber
\mathcal{A}_i=\{&\underbrace{y_i,y_{i+\tau},\ldots,y_{i+(j-1)\tau}}_{\text{Known~messages~to}~R_{k}},\\&
\nonumber
\underbrace{y_{i+j\tau}}_{\text{wanted message}},\underbrace{y_{i+(j+1)\tau}\ldots,y_{i+(j+\gamma-1)\tau}}_{\text{Interfernce~to}~R_{k}},\\&\underbrace{y_{i+(j+\gamma)(U+1)},y_{i+(j+\gamma+1)(U+1)},\ldots,y_{i+(t-1)\tau}}_{\text{Known~messages~to}~R_{k}}\},
\end{align}
for $j \in [0:t-1]$. 

Let $\mathbf{L}$ be a $t \times \gamma$ AIR matrix. In $\mathbf{L}$, any set of $\gamma$ adjacent rows are linearly independent. Consider the $b(D+1)+a$ code symbols $c_0,c_1~c_2,\ldots,c_{b(D+1)+a-1}$ given in \eqref{novcmdsv1}.
\begin{align}
\label{novcmdsv1}
\nonumber
&[c_i~c_{i+\tau}~\ldots~c_{i+(\gamma-1)\tau}]=\\&[y_i~y_{i+\tau}~y_{i+2\tau}~\ldots~y_{i+(t-1)\tau}]\mathbf{L}_{t \times \gamma}
\end{align}
for $i \in [0:\tau-1]$.

We prove that receiver $R_k$ decodes its wanted messages $x_{k,j}$ for every $k \in [0:K-1]$ and $j \in [1:b]$ by using $\gamma$ index code symbols out of available $b(D+1)+a$ index code symbols given in \eqref{novcmdsv1}. Let $(kb+j)~\text{mod}~\tau=g_j$ and $\left \lfloor \frac{kb+j}{\tau}\right \rfloor=h_j$ for every $j \in [1:b]$. We have, $kb+j=g_j+h_j\tau$ and $y_{g_j+h_j\tau} \in \mathcal{A}_{g_j}$. Let $L_0,L_1,L_2,\ldots,L_{t-1}$ be the $t$ rows of $\mathbf{L}_{t \times \gamma}$. From \eqref{novcmdsv1}, we have
\begin{align} 
\label{novcmdsv2}
\nonumber
&[c_{g_j}~c_{g_j+\tau}~\ldots~c_{g_j+(\gamma-1)\tau}]^{\mathsf{T}}\\&
\nonumber
=\underbrace{y_{g_j}L_0+y_{g_j+\tau}L_1+\ldots,y_{g_j+(b-1)\tau}L_{b-1}}_{\text{Known~messages~to}~R_{k}}\\&
\nonumber
+\underbrace{y_{g_j+b\tau}L_b}_{\text{wanted~message~to}~R_k}
\\&
\nonumber
+\underbrace{y_{g_j+(b+1)\tau}L_{b+1}+\ldots+y_{g_j+(b+\gamma-1)\tau}L_{b+\gamma-1}}_{\text{Interfernce~to}~R_{k}}\\&
+\underbrace{y_{g_j+(b+\gamma)\tau}L_{b+\gamma}+\ldots+y_{g_j+(t-1)\tau}L_{t-1}}_{\text{Known~messages~to}~R_{k}}.
\end{align}

In $\mathcal{A}_a$, receiver $R_k$ knows $t-\gamma$ messages as shown in \eqref{novcmdsv3} and \eqref{novcmdsv2}. Hence, after subtracting the known information, \eqref{novcmdsv2} is a set of $\gamma$ equations in $\gamma$ unknowns. The $\gamma$ equations in \eqref{novcmdsv2} are independent follows from the fact that any $\gamma$ adjacent rows in $\mathbf{L}_{t \times \gamma}$ are linearly independent. Hence, $R_k$ can decode its wanted message $y_{g_j+b\tau}=x_{k,j}~\forall j$. The rate achieved by this method is $\frac{b(D+1)+a}{b}=D+1+\frac{a}{b}$.
\end{proof}
\begin{note}
In \cite{TRCR}, it is shown that the message probability of error in decoding a message at a particular receiver decreases with a decrease in the number of transmissions used to decode the message among the total of broadcast transmissions. The encoding and decoding method given in Theorem \ref{novcccthmv1} indicates that every receiver uses atmost $\frac{b(D+1)+a}{b(U+1)+c}=\gamma$ broadcast symbols to decode its wanted message symbol.
\end{note}
\begin{note}
Another application of the construction method given in Theorem \ref{novcccthmv1} is related to Instantly Decodable Network Coding (IDNC). IDNC deals with code designs when the receivers have no buffer and need to decode the wanted messages instantly without having stored previous transmissions. A recent survey article on IDNC with application to Device-to-Device (D2D) communications is \cite{DSAA}. These results are valid for index coding since it is a special case of network coding. In Theorem \ref{novcccthmv1}, if $\gamma=1$, then every receiver uses atmost one broadcast symbol to decode a message symbol and hence the code is instantly decodable. 
\end{note}

\begin{table*}[t]
\centering
\begin{tabular}{|c|c|c|c|c|c|c|c|c|}
\hline
$K$ &$D$&$U$&$a_{\text{min}}$&$b_{\text{min}}$&$D+1$&$D+1+\frac{a_{\text{min}}}{b_{\text{min}}}$&AIR matrix size& AIR matrix size \\
~ &~&~&~&~&(lowerbound&(upperbound&required in \cite{VaR5}& required in this paper \\
~ &~&~&~&~&on $\beta$)&on $\beta$)&~& by using partition \\
\hline
71 & 1 & 1&1&35&2&2.0285&$2485 \times 71$& {\color{blue}$35 \times 1^{**}$} \\
\hline
71 & 2 & 1,2&2&23&3&3.0869&$1633 \times 71$& {\color{blue}$23 \times 1^{**}$} \\
\hline
71 & 3 & 1&2&35&4&4.0571&$2485 \times 142$& $35 \times 2$ \\
\hline
71 & 3 & 2,3&5&17&4&4.1764&$1207 \times 71$& {\color{blue}$17 \times 1^{**}$} \\
\hline
71 & 4 & 1,2,3,4&1&14&5&5.0714&$994 \times 71$& {\color{blue}$14 \times 1^{**}$} \\
\hline
71 & 5 &1&3&35&6&6.0857&$2485 \times 71$& {\color{blue}$35 \times 1^{**}$} \\
\hline
71 & 5 & 2&4&23&6&6.1739&$1633 \times 142$& $23 \times 2$ \\
\hline
71 & 5 & 3,4,5&5&11&6&6.4545&$781 \times 71$& {\color{blue}$11 \times 1^{**}$} \\
\hline
71 & 6 &1,2,$\ldots$,6&1&10&7&7.1000&$710 \times 71$& {\color{blue}$10 \times 1^{**}$} \\
\hline
71 & 7 &1&4&35&8&8.1142&$2485 \times 284$& $35 \times 4$ \\
\hline
71 & 7 & 2,3&6&17&8&8.3529&$1207 \times142$&  $17 \times 2$\\
\hline
71 & 7 & 4,5,6,7&7&8&8&8.8750&$568 \times 71$& {\color{blue}$8 \times 1^{**}$} \\
\hline
71 & 8 &1&5&31&9&9.1612&$2201 \times 284$& $31 \times 4$  \\
\hline
71 & 8 & 2&6&23&9&9.2608&$1633 \times 213$& {\color{blue}$23 \times 1^{**}$} \\
\hline
71 & 8 & 3&7&15&9&9.4666&$1065 \times 142$& $15 \times 2$ \\
\hline
71 & 8 & 4,5,6,7,8&1&7&9&9.1428&$497 \times 71$& {\color{blue}$7 \times 1^{**}$} \\
\hline
71 & 9 &1,2,$\ldots$,9&1&7&10&10.1428&$497 \times 71$& {\color{blue}$7 \times 1^{**}$} \\
\hline
71 & 10 & 1&3&32&11&11.0937&$2272 \times 355$&  $32 \times 5$\\
\hline
71 & 10 & 2&4&19&11&11.2105&$1349 \times 213$& $19 \times 3$ \\
\hline
71 & 10 & 3,4,$\ldots$,10&5&6&11&11.8333&$426 \times 71$& {\color{blue}$6 \times 1^{**}$} \\
\hline
\end{tabular}
\vspace{10pt}
\caption{Partition Encoding for SUICP(SNI) with $K=71$ and $U \leq D \leq 10$ (**Instantly decodable index codes).}
\label{table1}
\end{table*}
\begin{example}
\label{novcccex5}
Consider a SUICP(SNI) with $K=13,D=4,U=1$. For this SUICP(SNI), we have $(a=1,b=5) \in \mathcal{\mathbf{S}}$ and in \cite{VaR5}, we showed that the rate $D+1+\frac{a}{b}=5.2$ can be achieved by using AIR matrix of size $65 \times 26$ and 5-dimensional vector linear index coding. However, by using the partition method given in this paper, in this example, we show that the AIR matrix of size $5 \times 2$ is adequate for achieving a rate of $D+1+\frac{a}{b}=5.2$.


For this SUICP(SNI), we have 
\begin{align*}
&t=\frac{Kb}{b(U+1)+c}=5,\\&
\gamma=\frac{b(D+1)+a}{b(U+1)+c}=2, \\&
\tau=b(U+1)+c=13.
\end{align*}
and
\begin{align*}
&\mathcal{A}_1=\{x_{0,1},x_{2,4},x_{5,2},x_{8,0},x_{10,3}\}, \\&
\mathcal{A}_2=\{x_{0,2},x_{2,5},x_{5,3},x_{8,1},x_{10,4}\}, \\&
\mathcal{A}_3=\{x_{0,3},x_{3,1},x_{5,4},x_{8,2},x_{10,5}\},\\&
\mathcal{A}_4=\{x_{0,4},x_{3,2},x_{5,5},x_{8,3},x_{11,1}\},\\&
\mathcal{A}_5=\{x_{0,5},x_{3,3},x_{6,1},x_{8,4},x_{11,2}\},\\&
\mathcal{A}_6=\{x_{1,1},x_{3,4},x_{6,2},x_{8,5},x_{11,3}\},\\&
\mathcal{A}_7=\{x_{1,2},x_{3,5},x_{6,3},x_{9,1},x_{11,4}\},\\&
\mathcal{A}_8~=\{x_{1,3},x_{4,1},x_{6,4},x_{9,2},x_{11,5}\},\\&
\mathcal{A}_9~=\{x_{1,4},x_{4,2},x_{6,5},x_{9,3},x_{12,1}\},\\&
\mathcal{A}_{10}=\{x_{1,5},x_{4,3},x_{7,1},x_{9,4},x_{12,2}\},\\&
\mathcal{A}_{11}=\{x_{2,1},x_{4,4},x_{7,2},x_{9,5},x_{12,3}\},\\&
\mathcal{A}_{12}=\{x_{2,2},x_{4,5},x_{7,3},x_{10,1},x_{12,4}\},\\&
\mathcal{A}_{13}=\{x_{2,3},x_{5,1},x_{7,4},x_{10,2},x_{12,5}\}.
\end{align*}

From the partition, in $\mathcal{A}_i$ for every $i \in [1:13]$, any receiver wanting a message in $\mathcal{A}_i$ knows three other consecutive messages in $\mathcal{A}_i$. To understand this, consider $i=7$. $\mathcal{A}_7$ comprises of one wanted message for each of the receivers $R_1,R_3,R_6,R_9$ and $R_{11}$. In $\mathcal{A}_7$, the side-information (SI) and interference of the respective receivers is given below:-
\begin{align*}
&\{\underbrace{x_{1,2}}_{\text{wanted message to}~R_1},\underbrace{x_{3,5}}_{\text{Interference to}~R_1},\underbrace{x_{6,3},x_{9,1},x_{11,4}}_{\text{SI to}~R_1}\},\\&
\{\underbrace{x_{1,2}}_{\text{SI to}~R_3},\underbrace{x_{3,5}}_{\text{wanted message to}~R_3},\underbrace{x_{6,3}}_{\text{interference to}~R_3},\underbrace{x_{9,1},x_{11,4}}_{\text{SI to}~R_3}\},\\&
\{\underbrace{x_{1,2},x_{3,5}}_{\text{SI to}~R_6},\underbrace{x_{6,3}}_{\text{wanted message to}~R_6},\underbrace{x_{9,1},}_{\text{interference to}~R_6},\underbrace{x_{11,4}}_{\text{SI to}~R_6}\},\\&
\{\underbrace{x_{1,2},x_{3,5},x_{6,3}}_{\text{SI to}~R_9},\underbrace{x_{9,1},}_{\text{wanted message to}~R_9},\underbrace{x_{11,4}}_{\text{interference to}~R_9}\},\\&
\{\underbrace{x_{1,2}}_{\text{interference to}~R_{11}},\underbrace{x_{3,5},x_{6,3},x_{9,1},}_{\text{SI to}~R_{11}},\underbrace{x_{11,4}}_{\text{wanted message to}~R_{11}}\}.
\end{align*}

In an AIR matrix of size $5 \times 2$, every two adjacent rows are linearly independent. Hence, AIR matrix of size $5 \times 2$ can be used as an encoding matrix for $\mathcal{A}_i$ for $i \in [1:13]$. AIR matrix of size $5 \times 2$ is given below.
\setlength\extrarowheight{-8pt}
\arraycolsep=0.9pt
$$\mathbf{L}_{5 \times 2}=\left[
\begin{array}{cc}
1 & 0  \\
0 & 1  \\
1 & 0  \\
0 & 1  \\
1 & 1  \\
 \end{array}
\right]$$

The $26$ broadcast symbols for this SUICP(SNI) is obtained by multiplying each of the $13$ partitions above with AIR matrix of size $5 \times 2$.
The $26$ broadcast symbols for this SUICP(SNI) is obtained as given below:
\begin{align*}
\nonumber
&[c_0~~c_{13}]=[x_{0,1}~x_{2,4}~x_{5,2}~x_{8,0}~x_{10,3}]\mathbf{L}_{5 \times 2}\\&
\nonumber
[c_1~~c_{14}]=[x_{0,2}~x_{2,5}~x_{5,3}~x_{8,1}~x_{10,4}]\mathbf{L}_{5 \times 2}\\&
\nonumber
[c_2~~c_{15}]=[x_{0,3}~x_{3,1}~x_{5,4}~x_{8,2}~x_{10,5}]\mathbf{L}_{5 \times 2} \\&
\nonumber
[c_3~~c_{16}]=[x_{0,4}~x_{3,2}~x_{5,5}~x_{8,3}~x_{11,1}]\mathbf{L}_{5 \times 2}\\&
\nonumber
[c_4~~c_{17}]=[x_{0,5}~x_{3,3}~x_{6,1}~x_{8,4}~x_{11,2}]\mathbf{L}_{5 \times 2}\\&
\nonumber
[c_5~~c_{18}]=[x_{1,1}~x_{3,4}~x_{6,2}~x_{8,5}~x_{11,3}]\mathbf{L}_{5 \times 2}\\&
\nonumber
[c_6~~c_{19}]=[x_{1,2}~x_{3,5}~x_{6,3}~x_{9,1}~x_{11,4}]\mathbf{L}_{5 \times 2}\\&
\nonumber
[c_7~~c_{20}]=[x_{1,3}~x_{4,1}~x_{6,4}~x_{9,2}~x_{11,5}]\mathbf{L}_{5 \times 2}\\&
\nonumber
[c_8~~c_{21}]=[x_{1,4}~x_{4,2}~x_{6,5}~x_{9,3}~x_{12,1}]\mathbf{L}_{5 \times 2}\\&
\nonumber
[c_9~~c_{22}]=[x_{1,5}~x_{4,3}~x_{7,1}~x_{9,4}~x_{12,2}]\mathbf{L}_{5 \times 2}\\&
\nonumber
[c_{10}~c_{23}]=[x_{2,1}~x_{4,4}~x_{7,2}~x_{9,5}~x_{12,3}]\mathbf{L}_{5 \times 2}\\&
\nonumber
[c_{11}~c_{24}]=[x_{2,2}~x_{4,5}~x_{7,3}~x_{10,1}~x_{12,4}]\mathbf{L}_{5 \times 2}\\&
\nonumber
[c_{12}~c_{25}]=[x_{2,3}~x_{5,1}~x_{7,4}~x_{10,2}~x_{12,5}]\mathbf{L}_{5 \times 2}.
\end{align*}

Let $k=3$. $R_3$ wants to decode $x_{3,1},x_{3,2},x_{3,3},x_{3,4}$ and $x_{3,5}$. We have $x_{3,j} \in \mathcal{A}_{2+j}$ for every $j \in [1:5]$. $R_3$ decodes $x_{3,j}$ from $[c_{1+j}~~c_{14+j}]$. In $\mathcal{A}_{2+j}$, $R_3$ knows three messages for every $j \in [1:5]$. Hence, after substituting the known messages, $R_3$ sees $[c_{1+j}~~c_{14+j}]$ as two equations with two unknowns and solves the wanted message $x_{3,j}$ for every $j \in [1:5]$. 
\end{example}

\begin{note}
In Example \ref{novcccex5}, the size of the AIR encoding matrix used is $5 \times 2$. Hence, the encoding and decoding method given in Theorem \ref{novcccthmv1} guarantees that every receiver uses atmost $2$ broadcast symbols to decode its wanted message symbol.
\end{note}

In \cite{VaR5}, we give an algorithm to compute the value of $a$ and $b$ in $\mathbf{S}$ such that $\frac{a}{b}$ is minimum. We refer these values of $a$ and $b$ as $a_{\text{min}}$ and $b_{\text{min}}$. We obtain an upperbound on the broadcast rate given by $R_{airm}=D+1+\frac{a_{\text{min}}}{b_{\text{min}}}$. We show that $R_{airm}$ coincides with the existing results on the exact value of the broadcast rate in the respective settings. For the given $(a,b) \in \mathbf{S}$, in Theorem \ref{novcccthmv1}, we constructed index codes with rate $R_{airm}=D+1+\frac{a}{b}$. Hence, the rate of index code construction given in Theorem \ref{novcccthmv1} coincides with existing results of the exact value of the broadcast rate in the respective settings.

\begin{example}
\label{ex2}
For SUICP(SNI) with $K=71$, $U \leq D \leq 10$, the upperbound on $\beta$ and lowerbound on $\beta$ are shown in Table \ref{table1}. For $U=D=1$, $R_{airm}=2.0285$ coincide with the reciprocal of capacity given by Blasiak \textit{et al.} in \eqref{cap5}. Note that the 9th column of Table \ref{table1} gives the maximum number of broadcast symbols used by any receiver to decode its wanted message. The 9th column of Table \ref{table1} also indicates instantly decodable codes. 
\end{example}

\section{Scalar linear index codes for SUICP(SNI)}
\label{sec6}
In this section, we give scalar linear index codes for SUICP(SNI) with arbitrary $K,D$ and $U$. In \cite{VaR6}, we proved Lemma \ref{lemma3} by proving the capacity of SUICP(SNI) with arbitrary $K,D$ and $U=\text{gcd}(K,D+1)-1$. This lemma is useful to derive Theorem \ref{thm10} of this section.
\begin{lemma}
\label{lemma3}
In the AIR matrix of size $m \times n$, every row $L_k$ is not in the span of $n-1$ rows above and $\text{gcd}(m,n)-1$ rows below $L_k$ for $k \in [0:m-1]$.
\end{lemma}	
The Theorem \ref{thm10} and Lemma \ref{lemma4},  we give a scalar linear index codes for SUICP(SNI). 
\begin{theorem}
\label{thm10}
Consider a SUICP(SNI) with arbitrary $K,D$ and $U$. Let $a$ and $b$ be the non negative integers satisfying the relation 
\begin{align}
\label{eq10}
gcd(K+a,D+1+a+b) \geq U+1+a.
\end{align}
Then, for this SUICP(SNI), an AIR matrix of size $(K+a) \times (D+1+a+b)$ can be used as an encoding matrix to generate an index code with length $D+1+a+b$.
\end{theorem}
\begin{proof}
Let $a$ and $b$ be the non negative integers satisfying \eqref{eq10}. Consider an SUICP(SNI) with $K+a$ messages $[\underbrace{x_0~x_1~\ldots~x_{K-1}}_{K~messages}~\underbrace{0~0~\ldots~0}_{a~zeros}]$ and $D+a+b$ interference above and $U+a$ interference before the desired message. We refer this SUICP(SNI) as modified SUICP(SNI). 

From Lemma \ref{lemma3}, in an AIR matrix of size $(K+a) \times (D+1+a+b)$, every row $L_k$ is not in the span of $D+a+b$ rows above and $\text{gcd}(K+a,D+a+b)-1 \geq U+a$ rows below $L_k$ for $k \in [0:K+a-1]$. Hence, AIR matrix of size $(K+a) \times (D+1+a+b)$ can be used as an encoding matrix to the modified SUICP(SNI).

A scalar linear index code of length $D+1+a+b$ generated by 
\begin{align*}
 [c_0~c_1~\ldots~c_{D+a+b}]=[\underbrace{x_0~x_1~\ldots~x_{K-1}}_{K~messages}~\underbrace{0~0~\ldots~0}_{a~zeros}]\mathbf{L},
\end{align*}
where $\mathbf{L}$ be a AIR matrix of size $(K+a) \times (D+1+a+b)$. 

A solution to this modified SUICP(SNI) can also be used as a solution to the SUICP(SNI) with $K,D$ and $U$. This completes the proof.
\end{proof}
\begin{example}
Consider SUICP(SNI) with $K=19,D=13$ and $U=3$. These $K,D$ and $U$ satisfy \eqref{eq10} with $a=1$ and $b=0$. Hence, AIR matrix of size $(K+a) \times (D+1+a+b)=20 \times 15$ can be used as an encoding matrix for this SUICP(SNI). The length of the index code is $15$.
\end{example}

\begin{example}
Consider SUICP(SNI) with $K=71,D=52$ and $U=16$. These $K,D$ and $U$ satisfy \eqref{eq10} with $a=1$ and $b=0$. Hence, AIR matrix of size $(K+a) \times (D+1+a+b)=72 \times 54$ can be used as an encoding matrix for this SUICP(SNI). he length of the index code is $54$.
\end{example}

The following lemma guarantees that the length of index code for SUICP(SNI) with arbitrary $K,D$ and $U$ is less than $D+U+1$.
\begin{lemma}
\label{lemma4}
For an SUICP(SNI) with arbitrary $K,D$ and $U$, AIR matrix of size $K \times (D+U+1)$ can be used as an encoding matrix.
\end{lemma}
\begin{proof}
From the definition of AIR matrix, in an AIR matrix of size $K \times (D+U+1)$, any $D+U+1$ adjacent rows are linearly independent. Let the scalar linear index code generated by AIR matrix $\mathbf{L}$ be
\begin{align}
\label{code}
 y=x\mathbf{L}=\sum_{k=0}^{K-1}x_kL_k.
\end{align}
From \eqref{code}, receiver $R_k$ wants to decode $x_k$ for every  $k \in [0:K-1]$. After substituting the $K-(D+U+1)$ side-information of $R_k$, \eqref{code} is a system of $D+U+1$ independent equations in $D+U+1$ unknowns. Hence, $R_k$ can decode $x_k$ for every $k \in [0:K-1]$.
\end{proof}

\section{Improved upper-bounds on the broadcast rate of SUICP(SNI)}
\label{sec7}
Let $\mathbf{S}=\{(a,b):\text{gcd}(bK,b(D+1)+a)\geq b(U+1)\}$. In Theorem \ref{novcccthmv1}, we gave reduced complexity index code for SUICP(SNI) with length $D+1+\frac{a}{b}$. In \cite{VaR5}, we gave an algorithm to find out the values of $a$ and $b$ in $\mathbf{S}$ such that $\frac{a}{b}$ is minimum and gave an index code with length $D+1+(\frac{a}{b})_{\text{min}}$. However, for a given $K$, for certain values of $D$ and $U$, we get $D+1+(\frac{a}{b})_{\text{min}}=K$ and this length does not give any advantage when compared with uncoded transmission. In Lemma \ref{lemmanew}, we give the values of $D$ and $U$ for a given $K$ for which $D+1+(\frac{a}{b})_{\text{min}}=K$. In \cite{VaR5}, we proved that $b \in [1:\left \lfloor \frac{K}{U+1} \right \rfloor]$. This fact is used in the proof of Lemma \ref{lemmanew}.

\begin{lemma}
\label{lemmanew}
Consider an SUICP(SNI) with $K,D$ and $U$. Let 
\begin{align}
\label{con1}
\mathcal{D}_l= \left[ \left \lfloor \frac{lK}{l+1} \right \rfloor :\left \lfloor \frac{(l+1)K}{l+2} \right \rfloor-1 \right],
\end{align}
\begin{align}
\label{con2}
\mathcal{U}_l=\left [ \left \lfloor \frac{K}{l+2} \right \rfloor :  \left \lfloor \frac{K}{l+1} \right \rfloor -1\right],
\end{align}
for some non negative integer $l$. Let $D$ and $U$ be such that 
\begin{align}
\label{con11}
D \in \mathcal{D}_l~~\text{and}~~~U \in \mathcal{U}_l
\end{align}
for $l \in \mathbb{Z}^+$. For this SUICP(SNI), we get $D+1+(\frac{a}{b})_{\text{min}}=K$.
\end{lemma}
\begin{proof}
Let $U \in \mathcal{U}_l$. We have $$b \in \{1,2,\ldots,\left \lfloor \frac{K}{U+1} \right \rfloor\}=\{1,2,\ldots,l+1\}.$$ Let $D=\left \lfloor \frac{lK}{l+1} \right \rfloor + t \in \mathcal{D}_l$ for some $t \in \mathbb{Z}^+$. We have
\begin{align*}
b(D+1)&=b(\left \lfloor \frac{lK}{l+1} \right \rfloor + t+1) > \frac{blK}{l+1}=b\left(\frac{l+1-1}{l+1}K \right)\\&=bK-\frac{bK}{l+1} 
\geq (b-1)K.
\end{align*} 

The last inequality in the above equation follows from the fact that $b \in [1:l+1]$. Let $b(D+1)=(b-1)K+\tilde{a}$ for $\tilde{a} \in \mathbb{Z}^+$. We have $\text{gcd}(Kb,b(D+1)+a)=\text{gcd}(Kb,K(b-1)+\tilde{a}+a)=K.$ 
This is possible only if $\tilde{a}+a=K$ (gcd of any two numbers $\alpha,\gamma~(\alpha <\gamma)$ is always less than or equal to  $\alpha-\gamma$). Hence $b(D+1)+a=Kb$ and $\frac{b(D+1)+a}{b}=K$. This completes the proof. 
\end{proof}

The sets $\mathcal{D}_l$ and $\mathcal{U}_l$ for $K=71$ and $l \in [1:5]$ are given in Table \ref{table4}. The value of $D+1+(\frac{a}{b})_{\text{min}}$ is given in Table \ref{table5} for 
\begin{itemize}
\item $D \in \mathcal{D}_l$ and $U \in \mathcal{U}_l$.
\item $D \in \mathcal{D}_l$ and $U \notin \mathcal{U}_l$.
\item $D \notin \mathcal{D}_l$ and $U \in \mathcal{U}_l$.
\item $D \notin \mathcal{D}_l$ and $U \notin \mathcal{U}_l$.
\end{itemize}
\begin{table*}
\centering
\setlength\extrarowheight{10pt}
\begin{tabular}{|c|c|c|}
\hline
$l$&$\mathcal{D}_l$ & $\mathcal{U}_l$\\
\hline
$1$&$\left[ \left \lfloor \frac{K}{2} \right \rfloor :\left \lfloor \frac{2K}{3} \right \rfloor-1 \right]=[35:46]$ & $\left [ \left \lfloor \frac{K}{3} \right \rfloor :  \left \lfloor \frac{K}{2} \right \rfloor -1\right]=[23:34]$ \\
\hline
$2$&$\left[ \left \lfloor \frac{2K}{3} \right \rfloor :\left \lfloor \frac{3K}{4} \right \rfloor-1 \right]=[47:52]$ & $\left [ \left \lfloor \frac{K}{4} \right \rfloor :  \left \lfloor \frac{K}{3} \right \rfloor -1\right]=[17:22]$ \\
\hline
$3$&$\left[ \left \lfloor \frac{3K}{4} \right \rfloor :\left \lfloor \frac{4K}{5} \right \rfloor-1 \right]=[53:55]$ & $\left [ \left \lfloor \frac{K}{5} \right \rfloor :  \left \lfloor \frac{K}{4} \right \rfloor -1\right]=[14:16]$ \\
\hline
$4$&$\left[ \left \lfloor \frac{4K}{5} \right \rfloor :\left \lfloor \frac{5K}{6} \right \rfloor-1 \right]=[56:58]$ &  $\left [ \left \lfloor \frac{K}{6} \right \rfloor :  \left \lfloor \frac{K}{5} \right \rfloor -1\right]=[11:13]$ \\
\hline
$5$&$\left[ \left \lfloor \frac{5K}{6} \right \rfloor :\left \lfloor \frac{6K}{7} \right \rfloor-1 \right]=[59:59]$&$\left [ \left \lfloor \frac{K}{7} \right \rfloor :  \left \lfloor \frac{K}{6} \right \rfloor -1\right]=[10:10]$ \\
\hline

\end{tabular}
\vspace{5pt}
\caption{}
\label{table4}
\vspace{-5pt}
\end{table*}

\begin{table*}
\centering
\setlength\extrarowheight{2.5pt}
\begin{tabular}{|c|c|c|c|c|c|c|c|}
\hline
$K$ &$D$&$U$&$a$&$b$&$D+1$& $D+1+(\frac{a}{b})_{\text{min}}$ & \text{Remark}\\
\hline
71 & 44 & 1,2,3,4,5&2&11&45&45.1818 &  $D \in \mathcal{D}_l$ and $U \notin \mathcal{U}_l$\\
\hline
71 & 44 & $6,7,\ldots,22$&7&3&45&47.3333 & $D \in \mathcal{D}_l$ and $U \notin \mathcal{U}_l$ \\
\hline
71 & 44 & 23,24,25,26&26&1&45&71& $D \in \mathcal{D}_l$ and $U \in \mathcal{U}_l$ \\
\hline
71 & 45 & 1,2&3&20&46&46.1500& $D \in \mathcal{D}_l$ and $U \notin \mathcal{U}_l$ \\
\hline
71 & 45 & $3,4,\ldots,22$&4&3&46&47.3333& $D \in \mathcal{D}_l$ and $U \notin \mathcal{U}_l$ \\
\hline
71 & 45 & 23,24,25,26&25&1&46&71& $D \in \mathcal{D}_l$ and $U \in \mathcal{U}_l$ \\
\hline
71 & 27 & 27 &15&2&28&35.5& $D \notin \mathcal{D}_l$ and $U \in \mathcal{U}_l$ \\
\hline
71 & 33 & 25&3&2&34&35.5& $D \notin \mathcal{D}_l$ and $U \in \mathcal{U}_l$ \\
\hline
71 & 15 & 2 &3&22&16&16.1363& $D \notin \mathcal{D}_l$ and $U \notin \mathcal{U}_l$ \\
\hline
71 & 3 &1&2&35&4&4.0571& $D \notin \mathcal{D}_l$ and $U \notin \mathcal{U}_l$ \\
\hline
\end{tabular}
\vspace{5pt}
\caption{$D+1+(\frac{a}{b})_{\text{min}}$ for $K=71$.}
\label{table5}
\vspace{-5pt}
\end{table*}

\begin{theorem}
\label{nub}
Let  $l_1=D+1+(\frac{a}{b})_{\text{min}}$, where $(\frac{a}{b})_{\text{min}}$ is the minimum value of $a$ and $b$ such that $(a,b) \in \mathbf{S}$ given in \eqref{set11} and $\frac{a}{b}$ is minimum. Let  $l_2=D+1+(a+b)_{\text{min}}$, where $(a+b)_{\text{min}}$ is the minimum value of $a$ and $b$ satisfying \eqref{eq10} and $(a+b)_{\text{min}}$ is minimum. The upperbound on the broadcast rate of SUICP(SNI) is given by 
\begin{align*}
\beta(G) \leq \text{min}(l_1,l_2,D+U+1).
\end{align*}
\end{theorem}
\begin{proof}
Proof follows from Theorem \ref{novcccthmv1}, Theorem \ref{thm10} and Lemma \ref{lemma4}.
\end{proof}
\begin{example}
Consider an SUICP(SNI) with $K=71,D=44$ and $U=23$. From Lemma \ref{lemmanew}, we have $l_1=71$. From Theorem \ref{thm10}, we have $l_2=71$. From Lemma \ref{lemma4}, we have $D+U+1=68$. Hence from Theorem \ref{nub}, the upperbound is given by
\begin{align*}
\beta \leq \text{min}(l_1,l_2,D+U+1)=68.
\end{align*} 
\end{example}
\begin{example}
Consider an SUICP(SNI) with $K=71,D=52$ and $U=16$. From Lemma \ref{lemmanew}, we have $l_1=71$. From Theorem \ref{thm10}, we have $l_2=54$. From Lemma \ref{lemma4}, we have $D+U+1=69$. Hence, from Theorem \ref{nub}, the upperbound is given by
\begin{align*}
\beta \leq \text{min}(l_1,l_2,D+U+1)=54.
\end{align*} 
\end{example}
\section{Discussion}
\label{sec4}
In this paper, we gave near-optimal vector linear index codes for SUICP(SNI) for receivers with small buffer size. We gave an improved upperbound on the broadcast rate of SUICP(SNI). The broadcast rate and optimal coding for SUICP(SNI) with arbitrary $K,D$ and $U$ is a challenging open problem.

\section*{Acknowledgment}
This work was supported partly by the Science and Engineering Research Board (SERB) of Department of Science and Technology (DST), Government of India, through J.C. Bose National Fellowship to B. Sundar Rajan.


\end{document}